\documentclass[]{llncs}

\usepackage{epsfig}
\usepackage{latexsym}
\usepackage{amsmath}
\usepackage{amsfonts}
\usepackage{amssymb}
\usepackage{graphicx}
\usepackage{epstopdf}

\newtheorem{fact}{Fact}

\begin{document}

\title{Reconfiguration of Dominating Sets}
 	\author{
		Akira Suzuki\inst{1}\thanks{Research supported by
  JSPS Grant-in-Aid for Scientific Research, Grant Number 24.3660.} \and
        Amer E. Mouawad\inst{2}$^{\star\star}$ \and
        Naomi Nishimura\inst{2}\thanks{Research supported by
  the Natural Science and Engineering Research Council of Canada.}
}
 \institute{Graduate School of Information Sciences, Tohoku University\\
 Aoba-yama 6-6-05, Aoba-ku, Sendai, 980-8579, Japan. \\
 	\email{a.suzuki@ecei.tohoku.ac.jp}
\and
David R. Cheriton School of Computer Science\\
University of Waterloo, Waterloo, Ontario, Canada.\\
\email{\{aabdomou, nishi\}@uwaterloo.ca}
 }
\maketitle

\begin{abstract}
  We explore a reconfiguration version of the dominating set problem,
  where a dominating set in a graph $G$ is a set $S$ of vertices such
  that each vertex is either in $S$ or has a neighbour in $S$.  In a
  reconfiguration problem, the goal is to determine whether there
  exists a sequence of feasible solutions connecting given feasible
  solutions $s$ and $t$ such that
  each pair of consecutive solutions is adjacent according to a
  specified adjacency relation. Two dominating sets
  are adjacent if one can be formed from the other by the addition or
  deletion of a single vertex.

  For various values of $k$, we consider properties of $D_k(G)$, the
  graph consisting of a vertex for each dominating set of size at most
  $k$ and edges specified by the adjacency relation.  Addressing an
  open question posed by Haas and Seyffarth, we demonstrate that
  $D_{\Gamma(G)+1}(G)$ is not necessarily connected, for $\Gamma(G)$ the maximum
  cardinality of a minimal dominating set in $G$. The result holds
  even when graphs are constrained to be planar, of bounded
  tree-width, or $b$-partite for $b \ge 3$.  Moreover, we construct an
  infinite family of graphs such that $D_{\gamma(G)+1}(G)$ has
  exponential diameter, for $\gamma(G)$ the minimum size of a
  dominating set.  On the positive side, we show that
  $D_{n-m}(G)$ is connected and of linear diameter for any graph $G$
  on $n$ vertices having at least $m+1$ independent edges.
\end{abstract}

\section{Introduction}
The {\em reconfiguration version} of a problem determines whether it
is possible to transform one feasible solution $s$ into a {\em target}
feasible solution $t$ in a step-by-step manner (a {\em reconfiguration})
such that each intermediate solution is also feasible.  
The study of such problems has received considerable attention
in recent literature \cite{FHHH11,GKMP09,IDHPSUU11,IKOZ12,KMM11} and is 
interesting for a variety of reasons. From an algorithmic standpoint, reconfiguration 
models dynamic situations in which we seek to transform a solution into
a more desirable one, maintaining feasibility during the process. Reconfiguration
also models questions of evolution; it can represent the evolution of a genotype
where only individual mutations are allowed and all genotypes must satisfy a
certain fitness threshold, i.e. be feasible.
Moreover, the study of reconfiguration yields insights into the
structure of the solution space of the underlying problem, crucial for the design of efficient
algorithms. In fact, one of the initial motivations behind such questions was to study
the performance of heuristics \cite{GKMP09} and random sampling methods \cite{CVJ08},
where connectivity and other properties of the solution space play a crucial role.
Even though reconfiguration gained popularity in the last decade or so, 
the notion of exploring the solution space of a given problem has been previously 
considered in numerous settings. One such example is the work of
Mayr and Plaxton~\cite{MP92}, where the authors consider the problem 
of transforming one minimum spanning tree of a weighted graph into another by 
a sequence of edge swaps.

Some of the problems for which the reconfiguration version 
has been studied include vertex
colouring~\cite{BB13temp,BC09,CVJ08,CVJ09,CVJ11}, list
edge-colouring~\cite{IKD12}, list L(2,1)-labeling~\cite{IKOZ12}, block puzzles~\cite{HD05}, independent
set~\cite{HD05,IDHPSUU11}, clique, set cover, integer programming,
matching, spanning tree, matroid bases~\cite{IDHPSUU11},
satisfiability~\cite{GKMP09}, shortest path~\cite{B12,KMM11}, subset sum~\cite{ID11},
dominating set~\cite{HS12,oursipec}, odd cycle transversal, feedback
vertex set, and hitting set~\cite{oursipec}.
For most $\mathbf{NP}$-complete problems, the reconfiguration version has
been shown to be $\mathbf{PSPACE}$-complete \cite{IDHPSUU11,IKD12,KMM12}, while
for some problems in $\mathbf{P}$, the reconfiguration question could be either in
$\mathbf{P}$ \cite{IDHPSUU11} or $\mathbf{PSPACE}$-complete \cite{B12}.

The problem of transforming input $s$ into input $t$ can be viewed as
the problem of determining if there is a path from $s$ to $t$ in a
graph representing feasible solutions.  
Such a path is called a {\em reconfiguration sequence}.
For the problem of dominating
set, the {\em $k$-dominating graph}, defined formally in
Section~\ref{sec-prelims}, consists of a node for each feasible
solution and an edge for each pair of solutions that differ by a
single vertex. Finding an $s$-$t$ path in this graph has been shown to be
$\mathbf{W[2]}$-hard~\cite{oursipec}, and hence not likely to yield even a
fixed-parameter tractable algorithm~\cite{DF97}.

Although having received less attention than the $s$-$t$ path problem,
other characteristics of the solution graph have been studied.
Determining the diameter of the reconfiguration graph will result in
an upper bound on the length of any reconfiguration sequence.  For a
problem such as colouring, one can determine the {\em mixing number},
the minimum number of colours needed to ensure that the entire graph
is connected; such a number has been obtained for the problem of list
edge-colouring on trees~\cite{IKD12}.

In previous work on reconfiguration of dominating sets, Haas and
Seyffarth~\cite{HS12} considered the connectivity of the graph of
solutions of size at most $k$, for various values of $k$ relative to
$n$, the number of vertices in the input graph $G$. They demonstrated
that the graph is connected when $k=n-1$ and $G$ has at least two
independent edges, or when $k$ is one greater than the maximum
cardinality of a minimal dominating set and $G$ is non-trivially
bipartite or chordal.  They left as an open question, answered
negatively here, whether the latter results could be extended to all
graphs.

In this paper we extend previous work by showing in
Section~\ref{sec-independent} that the solution graph is connected and of
linear diameter for $k = n-m$ for any input graph with at least $m+1$
independent edges, for any nonnegative integer $m$.
In Section~\ref{sec-counter}, we give a series of counterexamples
demonstrating that
$D_{\Gamma(G)+1}(G)$ is not guaranteed to be connected for planar graphs, graphs of
bounded treewidth, or $b$-partite graphs for $b \ge 3$.
In Section~\ref{sec-diameter}, we pose and
answer a question about the diameter of $D_{\gamma(G)+1}(G)$ by
showing that there is an infinite family of graphs of exponential
diameter.

\section{Preliminaries}\label{sec-prelims}
We assume that each $G$ is a simple, undirected graph on $n$ vertices
with vertex set $V(G)$ and edge set $E(G)$.
The {\em diameter} of $G$ is the maximum over all pairs
of vertices $u$ and $v$ in $V(G)$ of the length of the shortest path
between $u$ and $v$.

A set $S \subseteq V(G)$
is a {\em dominating set} of $G$ if and only if every vertex in $V(G)
\setminus S$ is adjacent to a vertex in $S$.
The minimum cardinality of any dominating set of $G$ is denoted by
$\gamma(G)$. Similarly,
$\Gamma(G)$ is the maximum cardinality of any minimal
dominating set in $G$.

For a vertex $u \in V(G)$ and a dominating set $S$ of $G$, we say $u$
is {\em dominated} by $v \in S$ if $u \notin S$ and $u$ is adjacent to
$v$. For a vertex $v$ in a dominating set $S$,
a {\em private neighbour} of $v$ is a vertex
dominated by $v$ and not dominated by any other vertex in $S$;
the {\em private neighbourhood of $v$} is the set of private neighbours of $v$.
A vertex $v$ in a
dominating set $S$ is {\em deletable} if $S \setminus \{ v \}$ is also
a dominating set of $G$.

\begin{fact}\label{fact-deletable}
A vertex $v$ is deletable if and only if $v$ has at least one neighbour in
$S$ and $v$ has no private neighbour.
\end{fact}

Given a graph $G$ and a positive integer $k$, we consider the {\em
$k$-dominating graph} of $G$, $D_k(G)$, such that each vertex in
$V(D_k(G))$ corresponds to a dominating set of $G$ of cardinality at
most $k$. Two vertices are adjacent in $D_k(G)$ if and only if the
corresponding dominating sets differ by either the addition or the
deletion of a single vertex; each such operation is a {\em reconfiguration step}.
Formally, if $A$ and $B$ are dominating
sets of $G$ of cardinality at most $k$, then there exists an edge
between $A$ and $B$ if and only if there exists a vertex $u \in V(G)$
such that $(A \setminus B) \cup (B \setminus A) = \{ u \}$.
We refer to vertices in $G$ using lower case letters (e.g. $u, v$) and to the vertices in
$D_k(G)$, and by extension their associated dominating sets, using
upper case letters (e.g. $A, B$).
We write $A \leftrightarrow B$ if there exists a path in $D_k(G)$
joining $A$ and $B$.  The following fact is a consequence of our
ability to add vertices as needed to form $B$ from $A$.

\begin{fact}\label{fact-subset}
If $A \subseteq B$, then $A
\leftrightarrow B$ and $B \leftrightarrow A$.
\end{fact}

\section{Graphs with $m+1$ independent edges}\label{sec-independent}
\begin{theorem}\label{th_domi_inde_01}
For any nonnegative integer $m$,
if $G$ has at least $m+1$ independent edges,
then $D_{n-m}(G)$ is connected for $n = |V(G)|$.
\end{theorem}

\begin{proof}
For $G$ a graph with $m+1$ independent edges $I = \{ \{u_i,w_i\}
\mid 0 \le i \le m\}$, we define $U = \{u_i \mid 0 \le i \le m\}$, $W = \{w_i
\mid 0 \le i \le m\}$, and the set of {\em outsiders} $R = V(G) \setminus (U \cup W)$.

Using any dominating set $S$ of $G$, we can partition $I$ as follows: edge
$\{u_i,w_i\}$, $0 \le i \le m$, is {\em clean} if neither $u_i$ nor
$w_i$ is in $S$, {\em $u$-odd} if $u_i \in S$ but $w_i \notin S$,
{\em $w$-odd} if $w_i \in S$ but $u_i \notin S$, {\em odd} if
$\{u_i,w_i\}$ is $u$-odd or $w$-odd, and {\em even} if $\{u_i,w_i\}
\subseteq S$.  We use $\textrm{clean}(S)$ and $\textrm{odd}(S)$,
respectively, to denote the numbers of clean and odd edges
for $S$. Similarly, we let $\textrm{$u$-odd}(S)$ and $\textrm{$w$-odd}(S)$
denote the numbers of $u$-odd and $w$-odd edges for $S$.  In the
example graph shown in Figure~\ref{figthm1}, $m + 1 = 7$ and $R =
\emptyset$. There is a single clean edge, namely $\{u_1, w_1\}$, three
$w$-odd edges, two $u$-odd edges, and a single even edge.

\begin{figure}[h]
\begin{centering}
\centerline{\includegraphics[scale=0.3]{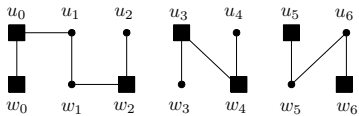}}
\end{centering}
\caption{Vertices in $S$ are marked with squares}
\label{figthm1}
\end{figure}

It suffices to show that for $S$ an arbitrary dominating set of $G$
such that $|S| \le n-m$, $S \leftrightarrow N$ for $N=V(G) \setminus
W$; $N$ is clearly a dominating set as each vertex $w_i \in W = V(G)
\setminus N$ is dominated by $u_i$. By Fact~\ref{fact-subset}, for
$S'$ a dominating set of $G$ such that $S' \supseteq S$ and $|S'| =
n-m$, since $S'$ is a superset of $S$, then $S \leftrightarrow S'$. The
reconfiguration from $S'$ to $N$ can be broken into three stages.
In the first stage, for a dominating set $S_0$ with no clean edges, we show $S'
\leftrightarrow S_0$ by
repeatedly decrementing the number of clean edges
($u_i$ or $w_i$ is added to the dominating set for some $0 \leq i \leq m$).
In the second stage,
for $T_m$ with $m$ $u$-odd edges and one even edge, we show
$S_0 \leftrightarrow T_m$ by repeatedly incrementing the number of $u$-odd edges.
Finally, we observe that deleting the single
remaining element in $T_m \cap W$ yields $T_m \leftrightarrow N$.

In stage 1, for $x = \textrm{clean}(S')$, we show that $S' = S_x
\leftrightarrow S_{x-1} \leftrightarrow S_{x-2} \leftrightarrow \ldots
\leftrightarrow S_{0} $ where for each $0 \le j \le x$, $S_j$ is a
dominating set of $G$ such that $|S_j|=n-m$ and $\textrm{clean}(S_j) =
j$. To show that $S_a \leftrightarrow S_{a-1}$ for arbitrary $1 \le a
\le x$, we prove that there is a deletable vertex
in some even edge and hence a vertex in a clean edge can be
added in the next reconfiguration step.  For $b = \textrm{odd}(S_a)$,
the set $E$ of vertices in even edges is of size
$2((m+1) - a - b)$. Since each vertex in $E$ has a neighbour in $S_a$, if 
 at least one vertex in $E$ does not have a private neighbour, then 
$E$ contains a deletable vertex (Fact~\ref{fact-deletable}).

The $m$ vertices in $V(G) \setminus S_a$ are the only possible
candidates to be private neighbours.
Of these, the $b$ vertices of $V(G) \setminus S_a$ in
odd edges cannot be private neighbours of vertices in $E$, as
each is the neighbour of a vertex in $S_a \setminus E$
(the other endpoint of the edge). The number of remaining candidates, $m-b$,
is smaller than the number of vertices in
$E$; $m \ge 2a + b$ as the vertices of
$V(G) \setminus S_a$ must contain both endpoints of any clean
edge and one endpoint for any odd edge.
Hence, there exists at least one deletable vertex in $E$.  When
we delete such a vertex and add an arbitrary endpoint of a
clean edge, the clean edge becomes an odd edge and the
number of clean edges decreases. We can therefore reconfigure
from $S_a$ to the desired dominating set, and by applying the
same argument $a$ times, to $S_0$.

In the second stage we show that
for $y = \textrm{$u$-odd}(S_0)$, $S_0 = T_y \leftrightarrow T_{y+1} \leftrightarrow T_{y+2} \leftrightarrow \ldots \leftrightarrow T_{m}$
where for each $y \le j \le m$, $T_j$ is a dominating set of $G$ such that
$|T_j| = n-m$, $\textrm{clean}(T_j) = 0$, and
$\textrm{$u$-odd}(T_j) = j$.  To show
that  $T_c \leftrightarrow T_{c+1}$ for arbitrary $y \le c \le m-1$, we use
a counting argument to find a vertex in an even
edge that is in $W$ and deletable; in one reconfiguration step the vertex is
deleted, increasing the number of $u$-odd edges, and in the
next reconfiguration step an arbitrary vertex in $R$ or in a $w$-odd edge is
added to the dominating set.
We let $d = \textrm{$w$-odd}(T_c)$ (i.e. the number of $w$-odd edges
for $T_c$) and observe that since there are $c$ $u$-odd edges,
$d$ $w$-odd edges, and no clean edges, there exist $(m+1)-c-d$
even edges.  We define $E_w$ to be the set of vertices in $W$
that are in the even edges, and observe that each has a
neighbour in $T_c$; a
vertex in $E_w$ will be deletable if it
does not have a private neighbour.

Of the $m$ vertices in $V(G)\backslash T_c$, only those in $R$
are candidates to be private neighbours 
of vertices in $E_w$, as each vertex in an
odd edge has 
a neighbour in $T_c$.  As there are $c$ $u$-odd edges and $d$ $w$-odd edges,
the total number of vertices in $R \cap V(G) \backslash T_c$
is $m - c - d$.  Since this is smaller than the number of vertices
in $E_w$, at least one vertex in $E_w$ must be deletable.
When we delete such a vertex from $T_c$ and in the next step
add an arbitrary vertex from the outsiders or $w$-odd edges,
the even edge becomes a $u$-odd edge and the number of
$u$-odd edges increases.  
Note that we can always find such a vertex 
since there are $m - c - d$ outsiders, $d$ $w$-odd edges, and $c \leq m - 1$.
Hence, we can reconfigure from
$T_c$ to $T_{c+1}$, and by $m - c$ repetitions, to $T_m$.
\qed
\end{proof}

Corollary~\ref{cor} results from the length of the
reconfiguration sequence formed in Theorem~\ref{th_domi_inde_01};
reconfiguring to $S'$ can be achieved in at most $n-m$ steps, and stages 1
and 2 require at most $2m$ steps each, as $m \in O(n)$ is at most the
numbers of clean and $u$-odd edges.
Theorem~\ref{th_domi_inde_02} shows that Theorem~\ref{th_domi_inde_01} is tight.

\begin{corollary}\label{cor}
The diameter of $D_{n-m}(G)$ is in $O(n)$ for $G$ a graph with $m+1$ independent edges.
\end{corollary}

\begin{theorem}\label{th_domi_inde_02}
For any nonnegative integer $m$, there exists a graph $G_m$ with $m$ independent edges such that $D_{n - m}(G_m)$ is not connected.
\end{theorem}

\begin{proof}
Let $G_m$ be a path on $n = 2m$ vertices.
Clearly, $G_m$ has $m$ disjoint edges, $n - m = 2m - m = m$, and $D_{n - m}(G_m) = D_{m}(G_m)$.
We let $S$ be a dominating set of $G_m$ such that $|S| \geq m + 1$.
At least one vertex in $S$ must have all its neighbors in $S$ and is therefore deletable.
It follows that $\Gamma(G_m) = m$ and $D_{n - m}(G_m) = D_{m}(G_m) = D_{\Gamma(G_m)}(G_m)$
which is not connected by the result of Haas and Seyffarth \cite[Lemma~3]{HS12}.
\qed
\end{proof}

\section{$D_{\Gamma(G) + 1}(G)$ may not be connected}\label{sec-counter}
In this section we demonstrate that $D_{\Gamma(G)+1}(G)$ is not
connected for an infinite family of graphs $G_{(d,b)}$ for all
positive integers $b \ge 3$ and $d \ge 2$, where graph $G_{(d,b)}$ is
constructed from $d+1$ cliques of size $b$.  We demonstrate using the
graph $G_{(4,3)}$ as shown in part (a) of Figure~\ref{fig1},
consisting of fifteen vertices partitioned into five cliques of size
3: the {\em outer clique} $C_0$, consisting of the top, left, and
right {\em outer vertices} $o_1$, $o_2$, and $o_3$, and the four {\em
  inner cliques} $C_1$ through $C_4$, ordered from left to right.  We
use $v_{(i,1)}$, $v_{(i,2)}$, and $v_{(i,3)}$ to denote the top, left,
and right vertices in clique $C_i$, $1 \le i \le 4$.  More generally,
a graph $G_{(d,b)}$ has $d+1$ $b$-cliques $C_i$ for $0 \le i \le
d$. The clique $C_0$ consists of outer vertices $o_j$ for $1 \le j \le
b$, and for each inner clique $C_i$, $1 \le i \le d$ and each $1 \le j
\le b$, there exists an edge $\{o_j,v_{(i,j)}\}$.

\begin{figure}[h]
\begin{centering}
\centerline{\includegraphics[scale=0.3]{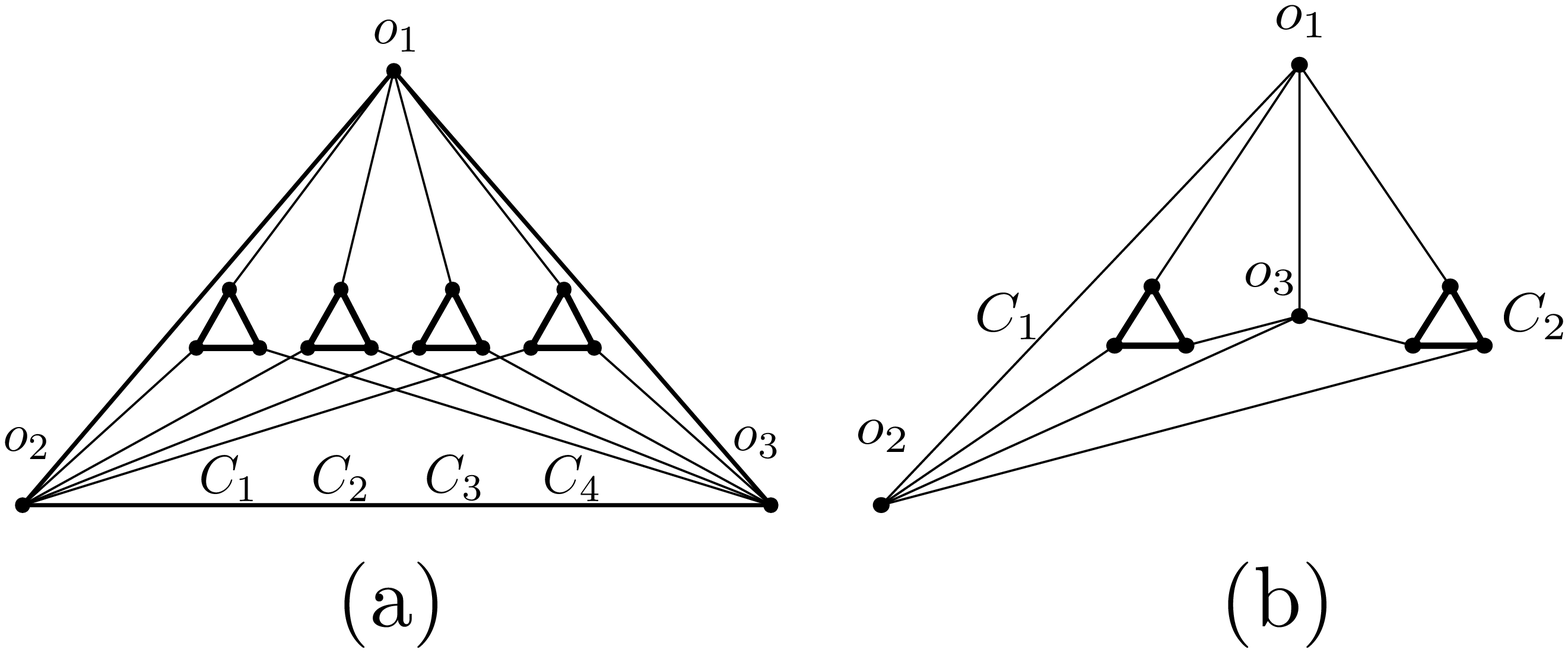}}
\end{centering}
\caption{Counterexamples for (a) general and (b) planar graphs}
\label{fig1}
\end{figure}

For any $1 \le j \le b$ a dominating set does not contain $o_j$, then
the vertices $v_{(i,j)}$ of the inner cliques must be dominated by
vertices in the inner cliques (hence Fact~\ref{obs-cover-inner}). In addition, the outer vertex $o_j$ can
be dominated only by another outer vertex or some vertex $v_{(i,j)}$,
$1 \le i \le d$ (hence Fact~\ref{obs-cover-outer}). 

\begin{fact}\label{obs-cover-inner}
Any dominating set that does not contain all of the outer vertices
must contain at least one vertex from each of the inner cliques.
\end{fact}

\begin{fact}\label{obs-cover-outer}
Any dominating set that does not contain any outer vertex must contain
at least one vertex of the form $v_{(\cdot,j)}$ for each $1 \le j \le b$.
\end{fact}

\begin{lemma}\label{lemma-minimal}
For each graph $G_{(d,b)}$ as defined above, $\Gamma(G_{(d,b)}) = d+b-2$.
\end{lemma}

\begin{proof}
We first demonstrate that there is a minimal dominating set of size
$d+b-2$, consisting of $\{v_{(1,j)} \mid 2 \le j \le b\} \cup
\{v_{(i,1)} \mid 2 \le i \le d\}$; the first set dominates $b-1$ of
the outer vertices and the first inner clique and the second set
dominates $o_1$ and the rest of the inner cliques. The dominating set
is minimal, as the removal of any vertex $v_{(1,j)}$, $2 \le j \le b$,
would leave vertex $o_j$ with no neighbour in the dominating set and
the removal of any $v_{(i,1)}$, $2 \le i \le d$, would leave
$\{v_{(i,j)} \mid 1 \le j \le b\}$ with no neighbour in the dominating
set.

By Fact~\ref{obs-cover-inner}, any dominating set that does not
contain all outer vertices must contain at least one vertex in each of the $d$
inner cliques.  Since the outer vertices form a minimal dominating
set, any other minimal dominating set must contain at least one vertex
from each of the inner cliques.

We now consider any dominating set $S$ of size at least $d+b-1$
containing one vertex for each inner clique and show that it is not
minimal.  If $S$ contains at least one outer vertex, we can find a
smaller dominating set by removing all but the outer vertex and one
vertex for each inner clique, yielding a total of $d + 1 < d + b - 1$
vertices (since $b \ge 3$). Now suppose that $S$ consists entirely of inner vertices;
by Fact~\ref{obs-cover-outer}, $S$ contains at least one
vertex of the form $v_{(\cdot,j)}$ for each $1 \le j \le b$.
Moreover, for at least one value $1 \le j' \le b$, there exists more than
one vertex of the form $v_{(\cdot,j')}$ as $d+b-1 > b$.  This allows us
to choose $b$ vertices of the form $v_{(\cdot,j)}$ for each $1 \le j
\le b$ that dominate at least two inner cliques as well as all outer
vertices.  By selecting one member of $S$ from each of the remaining
$d-2$ inner cliques, we form a dominating
set of size $d + b -2 < d + b -1$, proving that $S$ is not minimal.
\qed
\end{proof}

\begin{theorem}\label{th2}
There exists an infinite family of graphs such that for each $G$ in the family, $D_{\Gamma(G) + 1}(G)$ is not connected.
\end{theorem}

\begin{proof}
For any positive integers $b \ge 3$ and $d \ge 2$, we 
show that there is no path between dominating sets $A$ to $B$ in $D_{d+b-1}(G_{(d,b)})$, where $A$ consists of the vertices in
the outer clique and $B$ consists of $\{v_{(i,\ell)} \mid 1 \le i \le
d, 1 \le \ell \le b, i \equiv \ell \mbox{ (mod $b$)} \}$; 

By Fact~\ref{obs-cover-inner}, before we can delete any of the vertices in $A$,
we need to add one vertex from each of the inner cliques,
resulting in a dominating set of size $d+b = \Gamma(G_{(d,b)}) + 2$.
As there is no such vertex in our graph, there is no way to connect $A$ and $B$.
\qed
\end{proof}

Each graph $G_{(d,b)}$ constructed for Theorem~\ref{th2} is a $b$-partite graph; we
can partition the vertices into $b$ independent sets, where the $j$th set, $1 \le j
\le b$ is defined as $\{v_{(i,j)} \mid 1 \le i \le d\} \cup \{o_{i}  \mid 1 \le i \le d, i \equiv j+1
\mbox{ (mod $b$)} \}$.  Moreover, we can form a tree decomposition of 
width $2b - 1$ of $G_{(d,b)}$, for all positive integers $b \ge 3$ and $d \ge b$,  
by creating bags with the vertices of the inner cliques and
adding all outer vertices to each bag.

\begin{corollary}\label{cor-treewidth}
For every positive integer $b \ge 3$, there 
exists an infinite family of graphs of tree-width $2b - 1$ such that for each $G$ in the family, 
$D_{\Gamma(G)+1}(G)$ is not connected, and an infinite family of $b$-partite 
graphs such that for each $G$ in the family, $D_{\Gamma(G)+1}(G)$ is not connected.
\end{corollary}

Theorem~\ref{th2} does not preclude the possibility that when restricted to planar
graphs or any other graph class that excludes $G_{(d,b)}$, $D_{\Gamma(G)+1}(G)$ is connected. 
However, the next corollary follows directly from the fact that
$G_{(2,3)}$ is planar (part (b) of Figure~\ref{fig1}).

\begin{corollary}\label{cor-planar}
There exists a planar graph $G$ for which $D_{\Gamma(G)+1}(G)$ is not connected.
\end{corollary}

\section{On the diameter of $D_k(G)$}\label{sec-diameter}
In this section, we obtain a lower bound on the diameter of the
$k$-dominating graph of a family of graphs $G_n$.  We describe $G_n$ in terms of
several component subgraphs, each playing a role in forcing
the reconfiguration of dominating sets.


A {\em linkage gadget} (part (a), Figure~\ref{figlink}) consists of five vertices,
the {\em external vertices} (or endpoints) $e_1$ and $e_2$,
and the {\em internal vertices} $i_1$, $i_2$, and $i_3$. The external vertices are adjacent to
each internal vertex as well as to each other; the following results from
the internal vertices having degree two:

\begin{fact}\label{obs-linkage}
  In a linkage gadget, the minimum dominating sets of size one are
  $\{e_1\}$ and $\{e_2\}$. Any dominating set containing an internal
  vertex must contain at least two vertices.  Any dominating set in a
  graph containing $m$ vertex-disjoint linkage gadgets with all internal vertices
  having degree exactly two must contain at least one vertex in each linkage gadget.
\end{fact}

\begin{figure}[h]
\begin{centering}
\centerline{\includegraphics[scale=0.5]{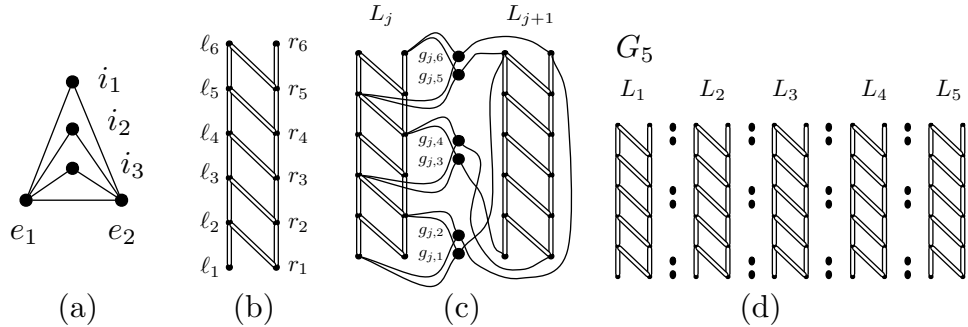}}
\end{centering}
\caption{Parts of the construction}
\label{figlink}
\end{figure}

A {\em ladder} (part (b) of Figure~\ref{figlink}, linkages shown as
double edges) is a graph consisting of twelve {\em ladder vertices}
paired into six {\em rungs}, where rung $i$ consists of the vertices
$\ell_i$ and $r_i$ for $1 \le i \le 6$, as well as the $45$ internal
vertices of fifteen linkage gadgets.  Each linkage gadget is
associated with a pair of ladder vertices, where the ladder vertices
are the external vertices in the linkage gadget.  The fifteen pairs
are as follows: ten {\em vertical pairs} $\{\ell_i,\ell_{i+1}\}$ and
$\{r_i,r_{i+1}\}$ for $1 \le i \le 5$, and five {\em cross pairs}
$\{\ell_{i+1},r_i\}$ for $1 \le i \le 5$.  For convenience, we refer
to vertices $\ell_i$, $1 \le i \le 6$ and the associated linkage
gadgets as the {\em left side of the ladder} and to vertices $r_i$, $1
\le i \le 6$ and the associated linkage gadgets as the {\em right side
  of the ladder}, or collectively as the {\em sides of the ladder}.

The graph $G_n$ consists of $n$ ladders $L_1$ through $L_n$ and $n-1$
sets of {\em gluing vertices}, where each set consists of
three {\em clusters} of two vertices each. For $\ell_{j,i}$ and $r_{j,i}$, $1
\le i \le 6$, the ladder vertices of ladder $L_j$, and 
$g_{j,1}$ through $g_{j,6}$ the gluing vertices that join
ladders $L_j$ and $L_{j+1}$, we have the following connections 
for $1 \le j \le n - 1$:

\begin{itemize}
\item{} Edges connecting the {\em bottom cluster} to the bottom two rungs of ladder $L_j$ and the top rung of ladder $L_{j+1}$: 
$\{\ell_{j,1}, g_{j,1}\}$, $\{\ell_{j,1}, g_{j,2}\}$, $\{r_{j,2}, g_{j,1}\}$, 
$\{r_{j,2}, g_{j,2}\}$, $\{\ell_{j+1,6}, g_{j,1}\}$, $\{r_{j+1,6}, g_{j,2}\}$

\item{} Edges connecting the {\em middle cluster} to the middle two rungs of ladder $L_j$ and the bottom rung of ladder $L_{j+1}$: 
$\{\ell_{j,3}, g_{j,3}\}$, $\{\ell_{j,3}, g_{j,4}\}$, $\{r_{j,4}, g_{j,3}\}$, 
$\{r_{j,4}, g_{j,4}\}$, $\{\ell_{j+1,1}, g_{j,3}\}$, $\{r_{j+1,1}, g_{j,4}\}$

\item{} Edges connecting the {\em top cluster} to the top two rungs of ladder $L_j$ and the top rung of ladder $L_{j+1}$: 
$\{\ell_{j,5}, g_{j,5}\}$, $\{\ell_{j,5}, g_{j,6}\}$, $\{r_{j,6}, g_{j,5}\}$, 
$\{r_{j,6}, g_{j,6}\}$, $\{\ell_{j+1,6}, g_{j,5}\}$, $\{r_{j+1,6}, g_{j,6}\}$
\end{itemize}

Figure~\ref{figlink} parts (c) and (d) show details of the
construction of $G_n$; they depict, respectively, two consecutive
ladders and $G_5$, both with linkages represented as double
edges. When clear from context, we sometimes use single subscripts
instead of double subscripts to refer to the vertices of a single
ladder.

We let ${\cal D} = \{\{\ell_{(j,2i-1)},\ell_{(j,2i)}\},\{r_{(j,2i-1)},$ $r_{(j,2i)}\} \mid 1
\le i \le 3, 1 \le j \le n\}$ denote a set of $6n$ pairs in $G_n$; the
corresponding linkage gadgets are vertex-disjoint.
Then Fact~\ref{obs-linkage} implies the following:

\begin{fact}\label{obs-ladder-six}
Any dominating set $S$ of $G_n$ must contain at least one vertex of
each of the linkage gadgets for vertical pairs in the set ${\cal D}$
and hence is of size at least $6n$; if $S$ contains an internal vertex, then $|S| > 6n$.
\end{fact}

Choosing an arbitrary external vertex for each vertical pair does not
guarantee that all vertices on the side of a ladder are
dominated; for example, the set $\{\ell_i \mid i \in \{1, 4, 5\}\}$
does not dominate the internal vertices in the vertical pair
$\{\ell_2,\ell_3\}$.  Choices that do not leave such gaps form the set ${\cal C} =
\{{\cal C}_i \mid 1 \le i \le 4\}$ where ${\cal C}_1 = \{1, 3, 5\}$, ${\cal C}_2 = \{2, 3, 5\}$,
${\cal C}_3 = \{2, 4, 5\}$, and ${\cal C}_4 = \{2, 4, 6\}$.

\begin{fact}\label{obs-minimum}
In any dominating set $S$ of size $6n$ and in any ladder $L$ in $G_n$,
the restriction of $S$ to $L$ must be of the form $S_i$ for some $1
\le i \le 7$, as illustrated in Figure~\ref{fig4}.
\end{fact}

\begin{figure}[h]
\begin{centering}
\centerline{\includegraphics[scale=0.3]{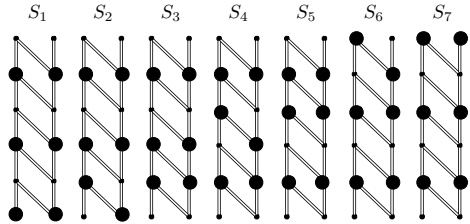}}
\end{centering}
\caption{Minimum dominating sets for $G_1$}
\label{fig4}
\end{figure}

\begin{proof}
Fact~\ref{obs-ladder-six} implies that the only choices for the
left (right) vertices are $\{\ell_i \mid i \in {\cal C}_j\}$ ($\{r_i \mid i
\in {\cal C}_j\}$) for $1 \le j \le 4$.  The sets $S_i$,
$ 1 \le i \le 7$, are the only combinations of these choices that
dominate all the internal vertices in the cross pairs.
\qed
\end{proof}

\noindent We say that ladder $L_j$ {\em is in
 state $S_i$} if the restriction of the dominating set to
$L_j$ is of the form $S_i$, for $1 \le j \le n$ and $1 \le i \le 7$.

The exponential lower bound in
Theorem~\ref{theorem-diameter} is based on counting how many times
each ladder is modified from $S_1$ to $S_7$ or vice versa; we say
ladder $L_j$ undergoes a {\em switch} for each such modification.
We first focus on a single ladder.

\begin{fact}\label{obs-deletable}
  For $S$ a dominating set of $G_1$, a vertex $v \in S$ is deletable
  if and only if either $v$ is the internal vertex of a linkage gadget
  one of whose external vertices is in $S$, or for every linkage
  gadget containing $v$ as an external vertex, either the other
  external vertex is also in $S$ or all internal vertices are in $S$.
\end{fact}

\begin{lemma}\label{lemma-tree}
In $D_{\gamma(G_1) + 1}(G_1)$ there is a single reconfiguration
sequence between $S_1$ and $S_7$, of length 12.
\end{lemma}

\begin{proof}
We define $P$ to be the path in the graph corresponding to the
reconfiguration sequence  $S_1 \leftrightarrow S_1 \cup
\{\ell_2\} \leftrightarrow S_2 \leftrightarrow S_2 \cup \{r_2\}
\leftrightarrow S_3 \leftrightarrow S_3 \cup \{\ell_4\}
\leftrightarrow S_4 \leftrightarrow S_4 \cup \{r_4\} \leftrightarrow
S_5 \leftrightarrow S_5 \cup \{\ell_6\} \leftrightarrow S_6
\leftrightarrow S_6 \cup \{r_6\} \leftrightarrow S_7$ and demonstrate
that there is no shorter path between $S_1$ and $S_7$.

By Facts~\ref{obs-minimum} and
\ref{obs-ladder-six}, $G_1$ has exactly seven dominating
sets of size six, and
any dominating set $S$ of size seven contains two vertices
from one vertical pair $d$ in ${\cal D}$ and one from each of the remaining five.
The neighbours of $S$ in $D_{\gamma(G_1)+1}(G_1)$ are the
vertices corresponding to the sets $S_i$, $1 \le i \le 7$,
obtained by deleting a single vertex of $S$.  The number of
neighbours is thus at most two, depending on which, if any,
vertices in $d$ are deletable.

If at least one of the vertices of $S$ in $d$ is an internal vertex,
then at most one vertex satisfies the first condition in
Fact~\ref{obs-deletable}.  Thus, for $S$ to have two neighbours, there must be a ladder vertex
that satisfies the second condition of
Fact~\ref{obs-deletable}, which by inspection of
Figure~\ref{fig4} can be seen to be false.

If instead $d$ contains two ladder vertices, in order for $S$ to have
two neighbours, the four ladder vertices on the side containing $d$
must correspond to the union of two of the sets in ${\cal C}$.  There are
only three such unions, ${\cal C}_1 \cup {\cal C}_2$, ${\cal C}_2 \cup {\cal C}_3$,
and ${\cal C}_3 \cup {\cal C}_4$, which implies that the only pairs with common neighbours
are $\{S_i,S_{i+1}\}$ for $1 \le i \le 6$, as needed to complete the
proof.
\qed
\end{proof}

For $n>2$, we cannot reconfigure ladders independently from each other,
as we need to ensure that all gluing vertices are dominated.  For
consecutive ladders $L_j$ and $L_{j+1}$, any cluster that is not
dominated by $L_j$ must be dominated by $L_{j+1}$; the bottom, middle,
and top clusters are not dominated by any vertex in $S_2$, $S_4$, and
$S_6$, respectively.

\begin{fact}\label{obs-two-ladders}
In any dominating set $S$ of $G_n$, for any $1 \le j < n$,
if $L_j$ is in state $S_2$, then $L_{j+1}$ is in state $S_7$;
if $L_j$ is in state $S_4$, then $L_{j+1}$ is in state $S_1$; and
if $L_j$ is in state $S_6$, then $L_{j+1}$ is in state $S_7$.
\end{fact}

\begin{lemma}\label{lemma-switch}
For any reconfiguration sequence in which $L_j$ and $L_{j+1}$ are
initially both in state $S_1$,
if $L_j$ undergoes $p$ switches then $L_{j+1}$ must undergo at
least $2p + 1$ switches.
\end{lemma}

\begin{proof}
We use a simple counting argument. When $p = 1$, the result
follows immediately from Fact~\ref{obs-two-ladders} since
$L_j$ can only reach state $S_7$ if $L_{j+1}$ is reconfigured
from $S_1$ to $S_7$ to $S_1$ and finally back to $S_7$.
After the first switch of $L_j$, both ladders are in state $S_7$.

For any subsequent switch of $L_j$, $L_j$ starts in state $S_7$
because for $L_j$ to reach $S_1$ from $S_2$ or to reach $S_7$ from
$S_6$, by Fact~\ref{obs-two-ladders} $L_{j+1}$ must have been in
$S_7$.  Since by definition $L_j$ starts in $S_1$ or $S_7$, to enable
$L_j$ to undergo a switch, $L_{j+1}$ will have to undergo at least two
switches, namely $S_7$ to $S_1$ and back to $S_7$.  \qed
\end{proof}

\begin{theorem}\label{theorem-diameter}
For $S$ a dominating set of $G_n$ such that every ladder of $G_n$ is
in state $S_1$ and $T$ a dominating set of $G_n$ such that
every ladder of $G_n$ is in state $S_7$, the length of any
reconfiguration sequence between $S$ and $T$ is at least $12(2^{n+1} -
n - 2)$.
\end{theorem}

\begin{proof}
We first observe that Lemma~\ref{lemma-tree} implies that the
switch of any ladder requires at least twelve reconfiguration
steps; since the vertex associated with a dominating set containing a gluing vertex will have
degree at most one in the $k$-dominating graph, there are no shortcuts formed.

To reconfigure from $S$ to $T$, ladder $L_1$ must undergo at least one
switch. By Lemma~\ref{lemma-switch}, ladder $L_2$ will undergo at least $3 = 2^2-1$
switches, hence $2^j-1$ switches for ladder $L_j$, $1 \le j \le n$.
Since each switch requires twelve steps, the total number
of steps is thus at least
 $12 \sum_{i = 1}^{n}{(2^i - 1)}  =  12(2^{n+1} - n - 2)$.
\qed
\end{proof}

\begin{corollary}
There exists an infinite family of graphs such that for each graph $G_n$
in the family, $D_{\gamma(G_n) + 1}(G_n)$ has diameter $\Omega(2^n)$.
\end{corollary}

\section{Conclusions and future work}
In answering Haas and Seyffarth's question concerning the connectivity
of $D_{k}(G)$ for general graphs and $k = \Gamma(G) + 1$, we have
demonstrated infinite families of planar, bounded treewidth, and
$b$-partite graphs for which the $k$-dominating graph is not
connected. It remains to be seen whether $k$-dominating graphs
are connected for graphs more general than 
non-trivially bipartite graphs or chordal graphs, 
and whether $D_{\Gamma(G)+2}(G)$ is connected for all graphs. 
It would also be useful to know if there 
is a value of $k$ for which $D_{k}(G)$ is
guaranteed not to have exponential diameter.
Interestingly, for our connectivity and diameter examples, incrementing the
size of the sets by one is sufficient to break the proofs.

\bibliography{shortrefs}

\begin{thebibliography}{10}

\bibitem{BB13temp}
Marthe Bonamy and Nicolas Bousquet.
\newblock Recoloring bounded treewidth graphs.
\newblock In {\em Proc. of the 7th Latin-American Algorithms, Graphs, and
  Optimization Symp.}, 2013.

\bibitem{B12}
P.~Bonsma.
\newblock The complexity of rerouting shortest paths.
\newblock In {\em Proc. of Mathematical Foundations of Computer Science}, pages
  222--233, 2012.

\bibitem{BC09}
P.~Bonsma and L.~Cereceda.
\newblock Finding paths between graph colourings: {PSPACE}-completeness and
  superpolynomial distances.
\newblock {\em Theor. Comput. Sci.}, 410(50):5215--5226, 2009.

\bibitem{CVJ08}
L.~Cereceda, J.~van~den Heuvel, and M.~Johnson.
\newblock Connectedness of the graph of vertex-colourings.
\newblock {\em Discrete Math.}, 308(56):913--919, 2008.

\bibitem{CVJ11}
L.~Cereceda, J.~van~den Heuvel, and M.~Johnson.
\newblock Finding paths between 3-colorings.
\newblock {\em J. of Graph Theory}, 67(1):69--82, 2011.

\bibitem{CVJ09}
Luis Cereceda, Jan van~den Heuvel, and Matthew Johnson.
\newblock Mixing 3-colourings in bipartite graphs.
\newblock {\em European J. of Combinatorics}, 30(7):1593--1606, 2009.

\bibitem{DF97}
R.~G. Downey and M.~R. Fellows.
\newblock {\em Parameterized complexity}.
\newblock Spring-Verlag, New York, 1997.

\bibitem{FHHH11}
Gerd Fricke, Sandra~Mitchell Hedetniemi, Stephen~T. Hedetniemi, and Kevin~R.
  Hutson.
\newblock $\gamma$-{G}raphs of {G}raphs.
\newblock {\em Discussiones Mathematicae Graph Theory}, 31(3):517--531, 2011.

\bibitem{GKMP09}
P.~Gopalan, P.~G. Kolaitis, E.~N. Maneva, and C.H. Papadimitriou.
\newblock The connectivity of boolean satisfiability: computational and
  structural dichotomies.
\newblock {\em SIAM J. on Computing}, 38(6):2330--2355, 2009.

\bibitem{HS12}
R.~Haas and K.~Seyffarth.
\newblock The $k$-{D}ominating {G}raph.
\newblock {\em Graphs and Combinatorics}, March 2013.
\newblock Online publication.

\bibitem{HD05}
R.~A. Hearn and E.~D. Demaine.
\newblock {PSPACE}-completeness of sliding-block puzzles and other problems
  through the nondeterministic constraint logic model of computation.
\newblock {\em Theor. Comput. Sci.}, 343(1-2):72--96, 2005.

\bibitem{ID11}
T.~Ito and E.~D. Demaine.
\newblock Approximability of the subset sum reconfiguration problem.
\newblock In {\em Proc. of the 8th Annual Conf. on Theory and Applications of
  Models of Computation}, pages 58--69, 2011.

\bibitem{IDHPSUU11}
T.~Ito, E.~D. Demaine, N.~J.~A. Harvey, C.~H. Papadimitriou, Martha Sideri,
  Ryuhei Uehara, and Yushi Uno.
\newblock On the complexity of reconfiguration problems.
\newblock {\em Theor. Comput. Sci.}, 412(12-14):1054--1065, 2011.

\bibitem{IKD12}
T.~Ito, M.~Kami\'{n}ski, and E.~D. Demaine.
\newblock Reconfiguration of list edge-colorings in a graph.
\newblock {\em Discrete Applied Math.}, 160(15):2199--2207, 2012.

\bibitem{IKOZ12}
T.~Ito, K.~Kawamura, H.~Ono, and X.~Zhou.
\newblock Reconfiguration of list {L(2,1)}-labelings in a graph.
\newblock In {\em Proc. of the 23rd Int. Symp. on Algorithms and Computation},
  pages 34--43, 2012.

\bibitem{KMM11}
M.~Kami\'{n}ski, P.~Medvedev, and M.~Milani\v{c}.
\newblock Shortest paths between shortest paths.
\newblock {\em Theor. Comput. Sci.}, 412(39):5205--5210, 2011.

\bibitem{KMM12}
M.~Kami\'{n}ski, P.~Medvedev, and M.~Milani\v{c}.
\newblock Complexity of independent set reconfigurability problems.
\newblock {\em Theor. Comput. Sci.}, 439:9--15, June 2012.

\bibitem{MP92}
E.~W. Mayr and C.~G. Plaxton.
\newblock On the spanning trees of weighted graphs.
\newblock {\em Combinatorica}, 12(4):433--447, 1992.

\bibitem{oursipec}
A.~E. Mouawad, N.~Nishimura, V.~Raman, N.~Simjour, and A.~Suzuki.
\newblock On the parameterized complexity of reconfiguration problems.
\newblock In {\em Proc. of the 8th Int. Symp. on Parameterized and Exact
  Computation}, pages 281--294, 2013.

\end{thebibliography}

\end{document}